\documentclass[pra,superscriptaddress,twocolumn,longbibliography,reprint]{revtex4-2}
\usepackage{amsmath,amsfonts,amssymb,graphics,graphicx,epsfig,color,times,natbib,amsthm,bm,makecell,mathrsfs}
\usepackage{booktabs}
\usepackage{dcolumn}
\usepackage{dsfont}
\usepackage{subfigure}
\usepackage{hyperref}
\usepackage{appendix}
\usepackage[english]{babel}
\usepackage{color}
\theoremstyle{plain}
\newtheorem{theorem}{Theorem}
\newtheorem{lemma}{Lemma}
\usepackage[ruled]{algorithm2e}

\hypersetup{colorlinks=true, citecolor=blue, urlcolor=blue, linkcolor=blue}

\begin{document}
	\date{\today}
	\title{SOS decomposition for general Bell inequalities  in two qubits systems and its application to quantum randomness}
	\author{Wen-Na Zhao}
	\affiliation{College of Science, China University of Petroleum, 266580 Qingdao, P.R. China.}
	\author{Youwang Xiao}
	\affiliation{College of Science, China University of Petroleum, 266580 Qingdao, P.R. China.}
	\author{Ming Li}\email{liming@upc.edu.cn.}
	\affiliation{College of Science, China University of Petroleum, 266580 Qingdao, P.R. China.}
	\author{Li Xu}
	\affiliation{College of Science, China University of Petroleum, 266580 Qingdao, P.R. China.}
	\author{Shao-Ming Fei}
	\affiliation{School of Mathematical Science, Capital Normal University, 100048, Beijing, China}
	\begin{abstract}
		Bell non-locality is closely related with device independent quantum randomness. In this paper, we present a kind of sum-of-squares (SOS) decomposition for general Bell inequalities in two qubits systems. By using the obtained SOS decomposition, we can then find the measurement operators associated  with the maximal violation of considered Bell inequality. We also practice the SOS decomposition method by considering the (generalized) Clauser-Horne-Shimony-Holt (CHSH) Bell inequality, the Elegant Bell inequality, the Gisin inequality and the Chained Bell inequality as examples. The corresponding SOS decompositions and the measurement operators that cause the maximum violation values of these Bell inequalities are derived, which are consistent with previous results.
We further discuss the device independent quantum randomness by using the SOS decompositions of Bell inequalities. We take the generalized CHSH inequality with the maximally entangled state and the Werner state that attaining the maximal violations as examples. Exact value or lower bound on the maximal guessing probability using the SOS decomposition are obtained. For Werner state, the lower bound can supply a much precise estimation of quantum randomness when $p$ tends to $1$.
	\end{abstract}
	\maketitle
	\section{introduction}
	Two of the most remarkable features of quantum theory are its intrinsic randomness and its non-local character.
	Local measurements of composite quantum systems lead to correlations that are incompatible with local hidden variable theory. This phenomenon is known as quantum non-locality and has been recognized as an important resource for quantum information tasks \cite{ref1}, such as quantum key distribution \cite{ref2}, communication complexity \cite{ref3}, randomness generation \cite{ref4}. In \cite{ref4111,ref4222}, they study the quantum correlations in the minimal scenario and structure of the set of quantum correlators using semidefinite programming.
	It is clearly demonstrated by the fact that measurements on quantum states may violate the so-called Bell inequalities \cite{ref5}.
	Randomness is an essential resource \cite{ref6,ref7,ref8,ref10,ref11,ref12} for various information processing tasks. However, the existence and widespread use of two random number generators such as Pseudo-RNG and True-RNG does not guarantee that the random numbers we generate are truly random \cite{ref13}. In fact, random numbers generated based on any classical process do not inherently contain true randomness, and the randomness we see is only due to the observer's incomplete understanding of the overall mechanism of the system. Mathematically it is just a probabilistic combination of some deterministic events. In classical theory, according to Newton's theorem, all physical processes are deterministic, while Bell criterion indicates that quantum theory contains endowed randomness. In one word the randomness in classical physical systems is apparent randomness, while the randomness generated in quantum theory is intrinsic randomness. This randomness persists even if we have the full knowledge of the preparation of the state of the system. Thus, such randomness does not rely on the lack of knowledge about the systems or complexity of the calculations \cite{ref14}.
	
We consider two qubits systems $H_A\otimes H_B$. Every Bell inequality corresponds to a Bell operator
\begin{gather}\label{gb}
\mathscr{B}= \sum_{x=1}^{n}\sum_{y=1}^{m}\alpha_{xy}A_x\otimes B_y,
\end{gather}
such that the violation is obtained as $\beta=tr[\mathscr{B}\rho]$. If the maximal violation achievable by using quantum resources (ie.the quantum bound) is $\beta_{\max}$, the shifted Bell operator is defined as $\gamma=\beta_{\max}\mathbb{I} - \mathscr{B}$, where $\mathbb{I}=I_A\otimes I_B, I_A$ and $I_B$ stands for the identity operators on the two subspace $H_A$ and $H_B$ respectively. Every shifted Bell operator is by construction positive semi-definite since $\langle \psi | \mathscr{B} | \psi \rangle \le \beta_{\max} $ for all $|\psi \rangle$. Imagine the shifted Bell operator admits a decomposition
	\begin{gather}\label{N1}
		\beta_{\max}\mathbb{I} - \mathscr{B} =\sum_{\lambda} P_{\lambda}^{\dagger} P_{\lambda},
	\end{gather}
	where each $ P_{\lambda} $ is a polynomial in the operators $ A_x$ and $B_y$. The decomposition of Eq.\eqref{N1} is called the sum-of-squares (SOS) decomposition of the shifted Bell operator and can be defined in a completely general way for other Bell operators \cite{ref15}. Significantly, SOS decompositions allow one to extract potentially useful information about the physical state and measurements used to achieve the maximal violation of the corresponding Bell inequality \cite{ref16,ref17,ref18,ref19,ref20}. Note that it is much difficult to find a SOS decomposition of any Bell operator, especially one that imposes strict constraints on maximal quantum violation.
	Once found, verifying that Eq.\eqref{N1} holds and that
	$ \langle \mathscr{B} \rangle \le \beta_{\max} $ usually involves only a few simple computations. That is, an SOS provides a simple certificate that $ \langle \mathscr{B} \rangle \le \beta_{\max} $.
	Furthermore, the search for optimal SOS can be cast as a series
	of semi-definite programs (SDP) that turns out to be simply the dual formulation \cite{ref21} of the SDP hierarchy introduced in \cite{ref22,ref23}. In \cite{ref16}, the authors have introduced two families of SOS decompositions for the Bell operators associated with the tilted Clauser-Horne-Shimony-Holt (CHSH) expressions.
	Finally, as shown in \cite{ref24}, an optimal SOS, i.e., one for which
	$ \langle \mathscr{B} \rangle \le \beta_{\max} $ is a tight bound, provides useful information about the optimal quantum strategy and can find an application in robust self-testing. In \cite{ref25}, the authors have presented a nice method by which an SOS decomposition of a type of Bell inequalities are derived.
	
	In this paper, we present a kind of SOS decomposition form for general Bell inequalities and derive the measurement operators associated with the maximal violations of considered Bell inequalities. We also practice the SOS decomposition method by considering the (generalized) CHSH Bell inequality, the Elegant Bell inequality, the Gisin inequality and the Chained Bell inequality as examples. The corresponding SOS decompositions and the measurement operators that bring about the maximum violation values of these Bell inequalities are derived, which are consistent with previous results. Compared with the result in \cite{ref16},
	the present SOS decomposition forms are more concise and applicable to general Bell inequalities.
	The SOS decomposition of general Bell inequalities is then applied to the calculation of the device-independent(DI) quantum randomness. We take the generalized CHSH Bell inequality together with the maximal entangled state and Werner state as examples. For the maximal entangled state we derive the exact value of the minimal entropy. While for Werner state, we analytically derive a lower bound on the maximal guessing probability by CHSH Bell inequality with the SOS decomposition method. Combining the SDP method in \cite{ref22,ref23}, we show that the lower bound can supply a much precise estimation of quantum randomness when $p$ tends to $1$.

	\section{SOS decomposition for general Bell inequalities and the optimal measurements}
	In this section, we present a kind of SOS decomposition for general Bell inequalities. Then by using the obtained SOS decomposition, we find the optimal measurement operators associated with the maximal violation of considered Bell inequality.

	\begin{theorem}\label{thm1}
Let $|\psi\rangle $ be a quantum  state and $\beta$ be a positive number. For any Bell operator in the form of (\ref{gb}), we define
\begin{align}\label{A2}
\omega_x=||\sum_{y=1}^{m}\alpha_{xy}I_A\otimes B_y|\psi\rangle||,
\end{align}
where$||\mathcal{O}|\psi\rangle||=\sqrt{\langle\psi|\mathcal{O}^+\mathcal{O}|\psi\rangle}.$
If $\sum_{x=1}^{n}\omega_{x}=\beta$, we can then obtain
\begin{align}\label{A1}
\langle\psi|\beta\mathcal{I}-\mathscr{B}|\psi\rangle
=\langle\psi|\sum_{x=1}^{n}\frac{\omega_x}{2}(M_x)^{\dag}M_x|\psi\rangle,		\end{align}
where
\begin{align}
M_x&=\frac{1}{\omega_x}(\sum_{y=1}^{m}\alpha_{xy}I_A\otimes B_y)-A_x\otimes I_B.
\end{align}
\end{theorem}

	\begin{proof}
		For any Bell operator in the form of (\ref{gb}), by substituting Eq.\eqref{A2} into Eq.\eqref{A1}, We can obtain
		\begin{widetext}\begin{align}\label{A3} &\langle\psi|\sum_{x=1}^{n}\frac{\omega_x}{2}(M_x)^{\dag}M_x|\psi\rangle\notag\\
&=\langle\psi|\sum_{x=1}^{n}\frac{\omega_x}{2}[\frac{1}{\omega_x}(\sum_{y=1}^{m}\alpha_{xy}I_A\otimes B_y)-A_x\otimes I_B]
[\frac{1}{\omega_x}(\sum_{y=1}^{m}\alpha_{xy}I_A\otimes B_y)-A_x\otimes I_B]\notag|\psi\rangle\\ &=\langle\psi|\sum_{x=1}^{n}\frac{\omega_x}{2}[\frac{1}{\omega_x^2}(\sum_{y=1}^{m}\alpha_{xy}I_A\otimes B_y)^2
-\frac{2}{\omega_x}A_x\otimes I_B(\sum_{y=1}^{m}\alpha_{xy}I_A\otimes B_y)+\mathbb{I}]\notag|\psi\rangle\\
&=\langle\psi|-\mathscr{B}|\psi\rangle
+\langle\psi|\sum_{x=1}^{n}\frac{1}{2\omega_x}(\sum_{y=1}^{m}\alpha_{xy}I_A\otimes B_y)^2|\psi\rangle
+\langle\psi|\sum_{x=1}^{n}\frac{1}{2\omega_x}\mathbb{I}|\psi\rangle\notag\\
&=\langle\psi|-\mathscr{B}|\psi\rangle+\langle\psi|\sum_{x=1}^{n}\omega_x\mathbb{I} |\psi\rangle =\langle\psi|-\mathscr{B}+\beta \mathbb{I}|\psi\rangle.
		\end{align}\end{widetext}
		\end{proof}
	The authors of the paper \cite{ref25} use the similar approach to give a tight upper bound on the expectation value for a class of Bell operators \eqref{eq1}, taking into account the measurements of Bob anti-commuting
		\begin{gather}	\label{eq1}	\mathscr{B}_n=\sum_{y=1}^{n}(\sum_{i=1}^{2^{n-1}}(-1)^{x_{y}^{i}}A_{n,i})\otimes B_{n,y}.
		\end{gather}
	In the following, we recover the measurement operator using the given state and maximum violation values. In particular, if we set $\beta=\beta_{\max}(|\psi\rangle)$ and $|\psi\rangle$ is the quantum state that corresponds to the $\beta_{\max}(|\psi\rangle)$, the above decomposition is the SOS decomposition of Bell inequality. Below we find the measurement operators corresponding to the  maximal violation of the Bell inequalities under these conditions.

	\subsection{The measurement operators corresponding to the maximal violation for any Bell inequality.}
For a given state $|\psi\rangle $ and a Bell inequality like Eq.\eqref{gb} and the associated upper bound $\beta_{\max}(|\psi\rangle)$, we can find the measurement operators corresponding to the upper bound. By Theorem \ref{thm1}, we find that $A_x, B_y$ is the measurement operators corresponding to the upper bound whenever they make the following two conditions satisfied,
\begin{align}
			&\beta_{\max}(|\psi\rangle) =\sum_{x=1}^{n}\omega_x ,\label{con1}\\
			&M_x|\psi \rangle = 0 \label{con2}.
\end{align}
Firstly, we get $B_y$ from Eq.\eqref{con1}. Then we combine Eq.\eqref{con2} with the obtained $B_y$ to represent the corresponding $A_x$,
\begin{align}
	A_x\otimes I_B=\frac{1}{\omega_x}(\sum_{y=1}^{m}\alpha_{xy}I_A\otimes B_y).
\end{align}
 In particular, when $|\psi\rangle $ is a maximally entangled state, $A_x$ can be represented linearly by $B_y$ as follows.
\begin{lemma}\label{lem1}
	For the maximum entangled state $ |\psi \rangle = \frac{1}{\sqrt{n}}\sum_{i=1}^{n}(| i \rangle_A |i \rangle_B)$ and a given symmetric matrix M with $M^{\mathrm{T}} =M$, we can obtain
	\begin{align}
	(I \otimes M -M \otimes I) |\psi \rangle =0.
	\end{align}
\end{lemma}
\begin{proof}
Let $M$ be a matrix with
	\begin{align}
		M=
		\begin{bmatrix}
			m_{11} & m_{12}  & \cdots   & m_{1n}   \\
			m_{21} & m_{22}  & \cdots   & m_{2n}  \\
			\vdots & \vdots  & \ddots   & \vdots  \\
			m_{n1} & m_{n2}  & \cdots\  & m_{nn}  
		\end{bmatrix},
	\end{align}
in which $m_{ij} = m_{ji}$.
\begin{align*}
		&(I \otimes M -M \otimes I) |\psi \rangle\\
		= 	&(I \otimes M -M \otimes I) \frac{1}{\sqrt{n}}\sum_{i=1}^{n}(|i\rangle _{A}|i\rangle_{B})\\
		=&\sqrt{\frac{1}{n}}\sum_{i=1}^{n}(|i\rangle_{A}\otimes M|i\rangle_{B}-M|i\rangle_{A} \otimes |i\rangle_{B})\\
		=&\sqrt{\frac{1}{n}}(\sum_{i,x=1}^{n}m_{xi}|ix\rangle-\sum_{i,x=1}^{n}m_{xi}|xi\rangle))\\
		=&\sqrt{\frac{1}{n}}[\sum_{i,x=1}^{n}(m_{xi}-m_{ix})|ix\rangle]\\
		=&0.
\end{align*} 
From $m_{ij} = m_{ji}$, we can obtain $(I \otimes M -M \otimes I) |\psi \rangle = 0$.
\end{proof}
By Lemma \ref*{lem1}, if $|\psi\rangle$ is a maximally entangled state, we set 
\begin{align}\label{mx}
		A_x=\frac{1}{\omega_x}(\sum_{y=1}^{m}\alpha_{xy} B_y),
\end{align}
then $ 	M_x|\psi \rangle = 0$ can be obtained. In other word the measurement operators on Alice's side can be represented by the measurement operators on Bob's side. Likewise, we can also obtain the measurements on Bob's side by the measurements on Alice's side. We firstly select operators $A_x$ fulfilling $\beta_{\max}(|\psi\rangle)=\sum_{y=1}^{m}\omega_{y}$, in which
 \begin{align}\label{A111}
 \omega_{y}=||\sum_{x=1}^{n}\alpha_{xy}A_x\otimes I_B|\psi\rangle||,
\end{align}
And same as the above, if $|\psi\rangle$ is a maximally entangled state, we set 
\begin{align}\label{my}
 B_y=\frac{1}{\omega_y}\sum_{x=1}^{n}\alpha_{xy}A_x,
\end{align}
then $ 	M_y|\psi \rangle = 0$ can be obtained, where $	M_y=I_A\otimes B_y-\frac{1}{\omega_y}(\sum_{x=1}^{n}\alpha_{xy}A_x\otimes I_B)$. In this way we can find the corresponding measurement operators $A_x,B_y$.
One computes
\begin{align}\label{A6} \omega_x&=[(\alpha_{x,1}^2+...+\alpha_{x,m}^2)+\sum_{i<j}\alpha_{x,i}\alpha_{x,j}
\langle\{B_i,B_j\}\rangle]^{\frac{1}{2}},\notag\\ 
\omega_y&=[(\alpha_{y,1}^2+...+\alpha_{y,n}^2)+\sum_{i<j}\alpha_{y,i}\alpha_{y,j}
\langle\{A_i,A_j\}\rangle]^{\frac{1}{2}}.
\end{align}
where $\langle\{B_i,B_j\}\rangle=B_iB_j+B_jB_i$ is the anti-commutator.

 From Eq.\eqref{A6}, one can find if $\sum_{i<j}\alpha_{x,i}\alpha_{x,j}
 \langle\{B_i,B_j\}\rangle = 0$(more specifically, $\langle\{B_i,B_j\}\rangle = 0$ for all $i<j$), we obtain $\omega_{x}= \frac{\beta_{\max}(|\psi\rangle)}{n}=\sum_{y=1}^{m}\alpha_{xy}^2$ for any $x$. The CHSH, the generalized CHSH and the EBI Bell inequality are the examples of this case.
 If $\sum_{i<j}\alpha_{x,i}\alpha_{x,j}
 \langle\{B_i,B_j\}\rangle \neq 0 $, we can also find the corresponding measurement operator, such as the Gisin inequality and the chained Bell inequality.
 The same discussion can be made for the above case of $\omega_{y}$.
 In the next subsection, we derive the measurements corresponding to the maximal violations by considering the above mentioned Bell inequalities.

	\subsection{Examples}

	\subsubsection{CHSH-Bell inequality}
	As one of the most famous Bell operators, the CHSH-Bell operator can be expressed as follows \cite{ref241}
	\begin{gather}
		\mathscr{B}_{CHSH}=A_0\otimes B_0+A_0\otimes B_1+A_1\otimes B_0-A_1\otimes B_1 .
	\end{gather}
	The SOS decomposition and the measurement operators corresponding to the  maximal violation of CHSH-Bell inequality can be derived below based on the above subsection. One first verifies
	\begin{align}\label{L1}
		\omega_{x} = \sqrt{\sum_{y=0}^{1}\alpha_{xy}^2} = \sqrt{2}, x=0,1.
	\end{align}
	We set $|\psi\rangle=\frac{1}{\sqrt{2}}(|00\rangle+|11\rangle)$. Note that $\beta_{\max}$ for CHSH inequality is $2\sqrt{2}$. From Eq.\eqref{A6} and Eq.\eqref{L1} one gets $\sum_{i<j}\alpha_{x,i}\alpha_{x,j}
	\langle\{B_i,B_j\}\rangle = 0 $. We can set $B_0=X,\ B_1=Z$, where X and Z refer to the Pauli matrix. One checks
	\begin{align*}
		\omega_0=||I_A\otimes(B_0+B_1)|\psi\rangle||=||I_A\otimes(X+Z)|\psi\rangle||=\sqrt{2} ,\\
		\omega_1=||I_A\otimes(B_0-B_1)|\psi\rangle||=||I_A\otimes(X-Z)|\psi\rangle||=\sqrt{2} .
	\end{align*}
By (\ref{mx}) we have
	\begin{gather}
	A_0=\frac{X+Z}{\sqrt{2}},A_1 = \frac{X-Z}{\sqrt{2}}.
\end{gather}
The SOS decomposition of CHSH-Bell inequality is
\begin{gather}	\gamma=\frac{1}{\sqrt{2}}[(I_A\otimes\frac{B_0+B_1}{\sqrt{2}}-A_0\otimes I_B)^2\\
	+(I_A\otimes\frac{B_0-B_1}{\sqrt{2}}-A_1\otimes I_B)^2].\notag
\end{gather}
	
	\subsubsection{Generalized CHSH-Bell inequality}
	A  more complex example named generalized CHSH-Bell operator is discussed below, which is expressed as \cite{ref26}
	\begin{gather}
		\mathscr{B}_\alpha=\alpha A_0\otimes B_0+\alpha A_0\otimes B_1+A_1\otimes B_0-A_1\otimes B_1 \label{G1}.
	\end{gather}
	Firstly, we calculate the sum of squares of the coefficients of the generalized CHSH-Bell inequality
	\begin{align}\label{B1}
	\omega_{y}=\sqrt{\sum_{x=0}^{1}\alpha_{xy}^2}=\sqrt{\alpha^2+1}\ (y=0,1) .
\end{align}
We set $|\psi\rangle=\frac{1}{\sqrt{2}}(|00\rangle+|11\rangle)$. From Eq.\eqref{B1} we have $\sum_{i<j}\alpha_{y,i}\alpha_{y,j}
\langle\{A_i,A_j\}\rangle = 0 $. So we can set $A_0=Z,\ A_1=X$. One computes
	\begin{gather}
		\omega_0=||(\alpha A_0+A_1)\otimes I_B|\psi\rangle||
		=\sqrt{\alpha^2+1} ,\notag\\
		\omega_1=||(\alpha A_0-A_1)\otimes I_B|\psi\rangle||
		=\sqrt{\alpha^2+1}\notag .
	\end{gather}
	So for the generalized CHSH-Bell inequality its SOS decomposition is as follows
	\begin{gather}
		\gamma=\frac{\sqrt{\alpha^2+1}}{2}[(\frac{\alpha A_0+A_1}{\sqrt{\alpha^2+1}}\otimes I_B-I_A\otimes B_0)^2\\
		+(\frac{\alpha A_0-A_1}{\sqrt{\alpha^2+1}}\otimes I_B-I_A\otimes B_1)^2] ,\notag
	\end{gather}
	and $B_0,B_1$ can be linearly represented by $A_0,A_1$ as
	\begin{gather}
		B_0=\frac{\alpha Z+X}{\sqrt{\alpha^2+1}}=\cos u Z +\sin u X  ,\notag \\
		B_1=\frac{\alpha Z-X}{\sqrt{\alpha^2+1}}=\cos u Z -\sin u X \label{G2} .
	\end{gather}
	This result is exactly the same as the one mentioned in the previous literature \cite{ref26}.
	\subsubsection{EBI inequality}
The so called Elegant Bell inequality is given in \cite{ref25}, with Bell operator
	\begin{align}
		S&=A_1\otimes B_1+A_2\otimes B_1+A_3\otimes B_1+A_1\otimes B_2\\\notag
		&-A_2\otimes B_2-A_3\otimes B_2-A_1\otimes B_3+A_2\otimes B_3\\\notag
		&-A_3\otimes B_3-A_1\otimes B_4-A_2\otimes B_4+A_3\otimes B_4 .
	\end{align}
	The significant increase in complexity compared to CHSH-Bell inequality and the generalize CHSH-Bell inequality is that the optional measurement operators for the Alice side are increased from two to three, and the optional measurement operators for the Bob side are increased from two to four. The same two coefficient sums of squares are calculated to be
	\begin{align}\label{L2}
		\omega_{y}=\sqrt{\sum_{x=1}^{3}\alpha_{xy}^2}=\sqrt{3}\ (y=1,2,3).
	\end{align}
We set $|\psi\rangle=\frac{1}{\sqrt{2}}(|00\rangle+|11\rangle)$. From Eq.\eqref{L2} it satisfies $\sum_{i<j}\alpha_{y,i}\alpha_{y,j}
\langle\{A_i,A_j\}\rangle = 0 $. So we can set $A_1=X,\ A_2=Y,\ A_3=Z$. One computes
	\begin{align}
		&\omega_1=||(A_1+A_2+A_3)\otimes I_B|\psi\rangle||=\sqrt{3},\notag\\
		&\omega_2=||(A_1-A_2-A_3)\otimes I_B|\psi\rangle||=\sqrt{3},\notag\\
		&\omega_3=||(-A_1+A_2-A_3)\otimes I_B|\psi\rangle||=\sqrt{3},\notag\\
		&\omega_4=||(-A_1-A_2+A_3)\otimes I_B|\psi\rangle||=\sqrt{3}\notag .
	\end{align}
	For the EBI inequality the SOS decomposition is
	\begin{align}
		\gamma&=\frac{\sqrt{3}}{2}[(\frac{A_1+A_2+A_3}{\sqrt{3}}\otimes I_B-I_A \otimes B_1)^2 \notag\\
&+(\frac{A_1-A_2-A_3}{\sqrt{3}}\otimes I_B-I_A \otimes B_2)^2\notag\\
		&+(\frac{-A_1+A_2-A_3}{\sqrt{3}}\otimes I_B-I_A \otimes B_3)^2\notag\\
		&+(\frac{-A_1-A_2+A_3}{\sqrt{3}}\otimes I_B-I_A \otimes B_4)^2] .
	\end{align}
And the measurement operators corresponding to the  maximal violation of CHSH Bell inequality are
\begin{align}
	A_1&=X,\ A_2=Y,\ A_3=Z, \\
	B_1&=\frac{X+Y+Z}{\sqrt{3}},\ B_2 =\frac{X-Y-Z}{\sqrt{3}},\\
	B_3&=\frac{-X+Y-Z}{\sqrt{3}},\ B_4 = \frac{-X-Y+Z}{\sqrt{3}}.
\end{align}
	\subsubsection{Gisin Bell inequality}
	Let $ \mathcal{G}_n $ denotes the Bell operator for the Gisin Bell inequality\cite{ref27},
	\begin{align}
		\mathcal{G}_n = \sum_{i=1}^{n}(\sum_{j=1}^{n+1-i}A_i\otimes B_j -\sum_{j=n+2-i}^{n}A_i\otimes B_j ),
	\end{align}
	where $A_i$ and $B_j$ are the Hermitian operators with eigenvalues $ \pm 1 $ acting on Hilbert space $\mathcal{H}$. In particular, it has been shown in \cite{ref27} that the classical bound for LHV models and the maximum quantum violation of $\mathcal{G}_n$ amount to
	\begin{align}
		&\mathcal{G}_n^{LHV}=[\frac{n^2+1}{2}],\\
		&\mathcal{G}_n^Q=2n\cos(\frac{\pi}{2n})/\sin(\frac{\pi}{n}),
	\end{align}
	where $[x]$ denotes the largest integer smaller or equal to $x$ and $\mathcal{G}_n^Q$ is realized with the maximally entangled state of two qubits
	\begin{align}
		|\phi \rangle =\frac{1}{\sqrt{2}}(|00\rangle + |11\rangle).
	\end{align}
	Now, let's set $n=3$ as an example of a SOS decomposition. One gets
	\begin{align}
		\mathcal{G}_3 &= (A_1+A_2+A_3)\otimes B_1\notag\\
		&+(A_1+A_2-A_3)\otimes B_2\notag\\
		&+(A_1-A_2-A_3)\otimes B_3,
	\end{align}
	and the maximum quantum violation of $\mathcal{G}_3$ is
	\begin{align}
		\mathcal{G}_3 = \frac{6\cos(\frac{\pi}{6})}{\sin(\frac{\pi}{3})}=6.
	\end{align}
	From Eq.\eqref{A111} and \eqref{A6} we can get
	\begin{align*}
	\omega_1^2 =&||(A_1+A_2+A_3)\otimes I_B|\phi\rangle||^2,\\
	\omega_2^2 =&||(A_1+A_2-A_3)\otimes I_B|\phi\rangle||^2,\\
	\omega_3^2 =&||(A_1-A_2-A_3)\otimes I_B|\phi\rangle||^2,\\
	&\omega_1^2=\omega_2^2=\omega_3^2,
	\end{align*}
	which meets $\langle\{A_1,A_2\}\rangle = \langle\{A_2,A_3\}\rangle=-\langle\{A_1,A_3\}\rangle$. Since
	\begin{align}
		\sum_{y=1}^{3}\omega_y &= 6,\\
		\omega_1=\omega_2 &=\omega_3=2,
	\end{align}
	without loss of generality we make
	\begin{align}
		&\{A_1,A_2\} = \{A_2,A_3\}=-\{A_1,A_3\}=I.\label{EQ2}
	\end{align}
	In order to find the measurement operators for Alice side that satisfy the condition, we set the following measurements of linear combination of Pauli matrices for Alice
	\begin{align}
		A_i=r_i \cdot \sigma,\label{EQ1}
	\end{align}
	where
	\begin{align}
		r_i = (a_i,b_i,c_i),\\
		\sigma = (X,Y,Z).
	\end{align}
	From Eq.\eqref{EQ1} we get $\{A_i ,A_j\}= 2r_i \cdot r_j I $. Then one has
	\begin{align}
		r_1 \cdot r_2 &= \frac{1}{2},\\
		r_2 \cdot r_3 &= \frac{1}{2},\\
		r_1 \cdot r_3 &=-\frac{1}{2}.
	\end{align}
	This suggests that $ \langle r_1,r_2\rangle= \langle r_2,r_3\rangle=\frac{\pi}{3},\langle r_1,r_3\rangle=\frac{2\pi}{3}$, where $\langle a,b\rangle$ means the angle between the vectors $a$ and $b$.
	That is, as long as the above conditions are satisfied we can obtain the corresponding SOS decomposition. Without loss of generality, we pick the following set of angles
	
	\begin{align}
		r_1 &= (\sin\frac{\pi}{3},0,\cos\frac{\pi}{3}),\\
		r_2 &= (\sin\frac{2\pi}{3},0,\cos\frac{2\pi}{3}),\\
		r_3 &= (\sin\pi,0,\cos\pi).
	\end{align}
	The measurement operators can be set as
	\begin{align}
		A_1 &= \frac{\sqrt{3}}{2}X+\frac{1}{2}Z,\\
		A_2 &=  \frac{\sqrt{3}}{2}X-\frac{1}{2}Z,\\
		A_3 &= -Z.
	\end{align}
	And the SOS decomposition of the $\mathcal{G}_3$ is
	\begin{align}
		\gamma &=(\frac{A_1+A_2+A_3}{2}\otimes I_B-I_A\otimes B_1)^2\notag\\
		&+(\frac{A_1+A_2-A_3}{2}\otimes I_B-I_A\otimes B_2)^2\notag\\
		&+(\frac{A_1-A_2-A_3}{2}\otimes I_B-I_A\otimes B_3)^2.
	\end{align}
	from which $B_1,B_2,B_3$ can be linearly represented by $A_1,A_2,A_3$
	\begin{align}
		B_1&=\frac{A_1+A_2+A_3}{2}=-\frac{\sqrt{3}}{2}X+\frac{1}{2}Z,\\
		B_2&=\frac{A_1+A_2-A_3}{2}= -\frac{\sqrt{3}}{2}X-\frac{1}{2}Z,\\
		B_3&=\frac{A_1-A_2-A_3}{2}=-Z.
	\end{align}
	This result is consistent with that in \cite{ref27}.
	\subsubsection{chained Bell inequality}
	Another example is chained Bell inequality\cite{ref28}. Let $\mathcal{C}_n$ denotes the Bell operator for the chained Bell inequality with $n$ measurements on each party
	\begin{align}
		\mathcal{C}_n = \sum_{n-1}^{k}A_k\otimes B_k+A_{k+1}\otimes B_k+A_n\otimes B_n-A_1\otimes B_n.
	\end{align}
	The classical bound for LHV models and the maximum quantum value of $\mathcal{C}_n$ are given, respectively, by the author of Ref.\cite{ref28}
	\begin{align}
		&\mathcal{C}_n^{LHV}=2n-2,\\
		&\mathcal{C}_n^Q= 2n\cos\frac{\pi}{2n}.
	\end{align}
	We can compute
	\begin{align*}
	\omega_k^2=&||(A_k+A_{k+1})\otimes I_B|\psi\rangle||^2,\\
	\omega_n^2=&||(A_n-A_1)\otimes I_B|\psi\rangle||^2,\\
	\omega_k^2=&\omega_n^2,
	\end{align*}
	in which $|\psi\rangle=\frac{1}{\sqrt{2}}(|00\rangle+|11\rangle)$ and $k = 1,2\dots n-1$. We can obtain $\langle\{A_k,A_{k+1}\}\rangle =-\langle\{A_1,A_n\}\rangle$. Since
	\begin{align*}
		\sum_{y=1}^{n}\omega_y&= 2n\cos\frac{\pi}{2n},\\
		\omega_1= \omega_2&=\dots=\omega_n=2\cos\frac{\pi}{2n},
	\end{align*}
	without loss of generality we set
	\begin{align}
		\{A_k,A_{k+1}\} &=-\{A_1,A_n\}=2\cos\frac{\pi}{n}I,\label{EQ7}\\
		\omega_k^2=&\ \omega_n^2=2+2\cos\frac{\pi}{n}\notag.
	\end{align}
	In order to find the measurement operators for Alice side that satisfy the condition, we set the following measurements of linear combination of Pauli matrices for Alice
	\begin{align}
		A_i=r_i \cdot \sigma\label{EQ8},
	\end{align}
	where
	\begin{align}
		r_i = (a_i,b_i,c_i),\\
		\sigma = (X,Y,Z).
	\end{align}
	From Eq.\eqref{EQ8}, we have $\{A_i ,A_j\}= 2r_i \cdot r_j I $. Together with Eq.\eqref{EQ7}, one gets
	\begin{align}
		r_k \cdot r_{k+1} &= \cos\frac{\pi}{n},\\
		r_1 \cdot r_n &=\cos\frac{(k-1)\pi}{n},\notag\\
		&=-\cos\frac{\pi}{n}.
	\end{align}
	This suggests that $ \langle r_k,r_{k+1} \rangle =\frac{\pi}{n} ,\langle r_1,r_n \rangle=\frac{(n-1)\pi}{n}$.
	That is, as long as the above conditions are satisfied we can obtain the corresponding SOS decomposition. Without loss of generality, we pick the following set of angles
	\begin{align}
		r_1 &= (\sin0,0,\cos0),\\
		r_i &= (\sin\frac{(i-1)\pi}{n},0,\cos\frac{(i-1)\pi}{n}),\\
	\end{align}
	in which $i=1,2 \dots n$. The measurement operator is denoted by
	\begin{align}
		A_i=\sin\frac{(i-1)\pi}{n}X+\frac{(i-1)\pi}{n}Z.
	\end{align}
	And the SOS decomposition of $\mathcal{C}_n$ is
	\begin{align*}
		\gamma =& \cos\frac{\pi}{2n}[\sum_{k=1}^{n-1}(\frac{A_k+A_{k+1}}{2\cos\frac{\pi}{2n}}\otimes I_B-I_A \otimes B_k)^2\notag\\
+&(\frac{A_n-A_1}{2\cos\frac{\pi}{2n}}\otimes I_B-I_A \otimes B_n)^2],
	\end{align*}
	from which $B_j$ can be linearly represented by $A_i$ as follows
	\begin{align*}
		B_k=&\frac{A_k+A_{k+1}}{2\cos\frac{\pi}{2n}}\notag\\
		=&\sin\frac{(2k-1)\pi}{2n}X+\cos\frac{(2k-1)\pi}{2n}Z, 1\le k \le n-1,\\
		B_n=&\frac{A_n-A_1}{2\cos\frac{\pi}{2n}}\notag\\
		=&\sin\frac{(2n-1)\pi}{2n}X+\cos\frac{(2n-1)\pi}{2n}Z.
	\end{align*}
	Combining the above equations we obtain for any $1\le k \le n$
	\begin{align}
		B_k=&\frac{A_k+A_{k+1}}{2\cos\frac{\pi}{2n}}\notag\\
		=&\sin\frac{(2k-1)\pi}{2n}X+\cos\frac{(2k-1)\pi}{2n}Z.
	\end{align}
	These results are consistent with that in \cite{ref29}.
	\section{Generation of certified randomness from the generalized CHSH-bell inequality}
	As mentioned earlier, there is a set of quantum relations that are satisfied in quantum theory: $\langle\mathscr{B}\rangle_L<\langle\mathscr{B}\rangle_Q \le \langle\mathscr{B}\rangle_{\max}$. Here throughout this document, $\langle\mathscr{B}\rangle_L ,\ \langle\mathscr{B}\rangle_Q ,\ \langle\mathscr{B}\rangle_{\max}$ denote the classical maximum, the actual value, and the quantum maximum violated value of a particular Bell inequality. This shows that quantum correlations violating Bell inequality cannot be reproduced by prearranged classical strategy, which demonstrates the inherent randomness of the observed statistics. Taking the generalized CHSH-Bell inequality as an example, in conjunction with our previously proposed SOS decomposition, we investigate the relation between the randomness of a particular state generation and variable $\alpha$.\\\indent
	Multiple methods exist in quantum information for calculating entropy, such as the Shannon-entropy and von Neumann-entropy. However, the quantitative research on randomness mainly uses the minimum entropy \cite{ref26,ref30}, which we also use in this paper to compare with the previous research. We know that the minimum entropy is determined only by the maximum probability of incidence among all events, quantifying the minimum predictability among all probability distributions. Therefore, we believe that the minimum entropy provides the safest bound for the generation of randomness, which we now refer to as the guaranteed randomness bound. Below for a given Bell violation $\langle\mathscr{B}\rangle_L < \langle\mathscr{B}\rangle_Q $ we quantify the magnitude of the device-independent generated randomness $R_{\min}$ by means of minimal entropy \cite{ref31}\cite{ref32}
	\begin{align}
		R_{\min}=&\min\limits_{\vec{P}_{obs}}H_\infty(a,b|A_x,B_y)\notag\\
		=&-\log_2[\max\limits_{\vec{P}_{obs}}p(a,b|A_x,B_y)],\notag\\
		\rm{s.t.}\
		(&\rm{i})\ \vec{P}_{obs}\subset\{p(a,b|A_x,B_y)\},\notag\\
		(&\rm{ii})\ p(a,b|A_x,B_y)=Tr[\rho_{(AB)}(\Pi^{a}_{A_x}\otimes\Pi^b_{B_y})],\notag\\
		(&\rm{iii})\ \langle\mathscr{B}\rangle_Q=\sum_{x=1}^{n}\sum_{y=1}^{m}\langle\alpha_{xy}A_x\otimes B_y\rangle, \notag\\
		&\qquad \langle\mathscr{B}\rangle_L < \langle\mathscr{B}\rangle_Q \le \langle\mathscr{B}\rangle_{\max}  ,
	\end{align}
	in which $\vec{P}_{obs}\subset R^{mn}$ denotes a vector in a real vector space of dimension $nm$ that represents the set of all observed probabilistic relations, and $\Pi^{a}_{A_x}$ denotes the projection operator for the eigenvector corresponding to the eigenvalue of the measurement operator $A_x$ that has eigenvalue $a$.
	From the Eq.\eqref{G1}, we obtain the SOS decomposition as
	\begin{gather}
		\gamma=\frac{\sqrt{\alpha^2+1}}{2}[(\frac{\alpha A_0+A_1}{\sqrt{\alpha^2+1}}-B_0)^2+(\frac{\alpha A_0-A_1}{\sqrt{\alpha^2+1}}-B_1)^2]  .
	\end{gather}
	When it reaches its maximum violation value $2\sqrt{\alpha^2+1}$, the measurement operators of Alice's side are anti-commutation. Without loss of generality, we choose $A_0=Z,A_1=X$, and $B_0,B_1$ can be linearly represented by $A_0,A_1$, which are shown in Eq.\eqref{G2}.
	At the reach of the maximum violation, its quantum state is the maximally entangled state $|\psi\rangle=\frac{1}{\sqrt{2}}(|00\rangle+|11\rangle)$, so we first discuss the randomness generation with respect to the variable $\alpha$ for the maximally entangled state.
	\subsection{Maximally entangled state}
	First of all the decomposition can be solved to obtain its corresponding eigenvectors with eigenvalues of $1,\ -1$ respectively as
	\begin{gather}
		A_0=Z=\begin{pmatrix}
			1&0\\
			0&-1
		\end{pmatrix},
		A_1=X=\begin{pmatrix}
			0&1\\
			1&0
		\end{pmatrix};\notag\\
		B_0=\cos u Z+\sin u X =\begin{pmatrix}
			\cos u&\sin u \\
			\sin u & -\cos u
		\end{pmatrix},\notag\\
		B_1=\cos u Z -\sin u X=\begin{pmatrix}
			\cos u & -\sin u\\
			-\sin  u &-\cos u
		\end{pmatrix}.\notag
	\end{gather}
Let $a=\cos u -1,b=\sin u $. The projection operators corresponding to the measurement operators for eigenvalues $1,\ -1$ can be obtained
	\begin{gather}
		\Pi_{A_0}^1=\frac{1}{2}\begin{pmatrix}
			1&1\\
			1&1
		\end{pmatrix} ,\
		\Pi_{A_0}^{-1}=\frac{1}{2}\begin{pmatrix}
			1&-1\\
			1&-1
		\end{pmatrix} ;\notag\\
		\Pi_{A_1}^1=\begin{pmatrix}
			1&0\\
			0&0
		\end{pmatrix} ,\
		\Pi_{A_1}^{-1}=\begin{pmatrix}
			0&0\\
			0&1
		\end{pmatrix} ;\notag\\
		\Pi_{B_0}^1=\frac{1}{a^2+b^2}\begin{pmatrix}
			b^2&-ab\\
			-ab&a^2
		\end{pmatrix} ,\
		\Pi_{B_0}^{-1}=\frac{1}{a^2+b^2}\begin{pmatrix}
			a^2&ab\\
			ab&b^2
		\end{pmatrix} ;\notag\\
		\Pi_{B_1}^1=\frac{1}{a^2+b^2}\begin{pmatrix}
			b^2&ab\\
			ab&a^2
		\end{pmatrix} ,\
		\Pi_{B_1}^{-1}=\frac{1}{a^2+b^2}\begin{pmatrix}
			a^2&-ab\\
			-ab&b^2
		\end{pmatrix}\notag .
	\end{gather}\\\indent
	In the next step one can calculate $p(i,j|A_x,B_y)=\rm{Tr}[\rho_{AB}(\Pi^{i}_{A_x}\otimes\Pi^j_{B_y})]$, where $\rho_{AB}=|\psi\rangle\langle\psi|;\ i,j\subset\{1,-1\};\ x,y\subset\{0,1\}$. There are 16 cases in total. The calculation can be found that there are duplicate terms in the these cases. It is required to solve the largest joint probability among the 16 joint probabilities, which can be translate to calculate the comparative analysis the magnitude value within the range of a, b
	\begin{align*}
		P_1(a,b)&=\frac{(a+b)^2}{4(a^2+b^2)}, P_2(a,b)=\frac{(a-b)^2}{4(a^2+b^2)},\\
		P_3(a,b)&=\frac{a^2}{4(a^2+b^2)},
		P_4(a,b)=\frac{b^2}{4(a^2+b^2)}\label{G3}  .
	\end{align*}
	Thus one can simplify the discussion to analyze the value of $T_1(a,b)=(a+b)^2,\ T_2(a,b)=(a-b)^2,\ T_3(a,b)=a^2,\ T_4(a,b)=b^2$. And $a =\cos u -1=\frac{\alpha-1}{\sqrt{\alpha^2+1}}\le 0,\ b=\sin u=\frac{1}{\sqrt{\alpha^2 +1}}>0$ can be attained for $u\in(0,\frac{\pi}{4}]$ when $\alpha \ge 1$, which gives $T_1(a,b)\le T_2(a,b)$. Then we compare $T_3(a,b),\ T_4(a,b)$ by
	\begin{gather}
		T_3(a,b)-T_4(a,b)=a^2-b^2=2\cos u(\cos u - 1)  .
	\end{gather}
	For $u \in (0,\frac{\pi}{4}] $, we can get $T_3(a,b) < T_4(a,b)$. Below we compare the relationship between  $T_2(a,b),\ T_4(a,b)$ as follows
	\begin{gather}
		T_2(a,b)-T_4(a,b)=(\cos u -1)(\cos u -\sin u)\label{B2} .
	\end{gather}
	Observing Eq.\eqref{B2}, we find that when $u\in (0,\frac{\pi}{4}],\ T_2(a,b)\le T_4(a,b)$.
	We substitute $a=\cos u -1,\ b=\sin u$ to Eq.\eqref{G3}
	\begin{gather}
		T_3(a,b)=\frac{b^2}{2(a^2+b^2)}=\frac{1+\cos u}{4} .
	\end{gather}
	Then, for the generalized CHSH-Bell inequality, when the fixed quantum state is the maximally entangled state and the measurement operator is the one corresponding to reaching the maximum violation value, the magnitude of the generated randomness with respect to the variable $\alpha$ can be expressed as a function of the following
	\begin{gather}
		\max\limits_{\vec{P}_{obs}}p(a,b|A_x,B_y)=\frac{1+\cos u}{4} ,
	\end{gather}
	where $\cos u=\frac{\alpha}{\sqrt{\alpha^2+1}}$, and
	at this point, $R_{\min}=-\log_2(\frac{1+\cos u}{4})$. It is easy to get that when $u=\frac{\pi}{4}$, $P_{\max}=\frac{1}{4}(1+\frac{1}{\sqrt{2}}), \ R_{\min}=-\log_ 2[\frac{1}{4}(1+\frac{1}{\sqrt{2}})]=
	1.2284$ bits, which corresponds to $\alpha=1$, i.e. the CHSH-Bell inequality. The result is the same as that in the article \cite{ref6}. And we can find when
$u\to 0, P_{\max}\to\frac{1}{2}, R_{\min}\to-\log_2[\frac{1}{2}]\to1 \rm{bits},$
i.e. the randomness obtained by the measurements of reaching the maximum violation of generalized CHSH-Bell inequality is minimal. In addition, the variation of randomness with respect to $\alpha$ is shown in FIG. \ref{P1}.
	\subsection{Werner state}
	Next we consider the Werner state, which is defined as follows
\begin{align}
	\rho(p)=p|\psi_-\rangle\langle\psi_-|+(1-p)\frac{I}{4},
\end{align}
	where $|\psi_-\rangle=\frac{1}{\sqrt{2}}(|01\rangle-|10\rangle), 0\le p \le 1 $. We compute $p(i,j|A_x,B_y)=tr[\rho(p)(\ \Pi^{i}_{A_x}\otimes\Pi^j_{B_y})]$. For a specific Bell inequality, we can note that the set of measurement operators for a pure state $|\psi\rangle$ and its corresponding Werner state $p|\psi_-\rangle\langle\psi_-|+(1-p)\frac{I}{4}$ to reach the maximum violation are the same. So in this section we have selected measurements that are also the same as $|\psi\rangle$. 16 outcomes can be obtained, most of which are repetitive. Thus we just need consider the following four joint probabilities
	\begin{gather}
		P_1(a,b)=\frac{1}{4}-\frac{pab}{2(a^2+b^2)},P_2(a,b)=\frac{1}{4}+\frac{pab}{2(a^2+b^2)},\notag\\ P_3(a,b)=\frac{1}{4}+\frac{p(a^2-b^2)}{4(a^2+b^2)},P_4(a,b)=\frac{1}{4}+\frac{p(b^2-a^2)}{2(a^2+b^2)} \label{Q1}.
	\end{gather}
	Since $ab<0$, we get
	$$P_1(a,b)=\frac{1}{4}-\frac{pab}{2(a^2+b^2)}>P_2(a,b)=\frac{1}{4}+\frac{pab}{2(a^2+b^2)},$$
    $$P_3(a,b)=\frac{1}{4}+\frac{p(a^2-b^2)}{4(a^2+b^2)}\le P_4(a,b)=\frac{1}{4}+\frac{p(b^2-a^2)}{2(a^2+b^2)}.$$ 
    Then we can get $P_1( a,b)=\frac{1}{4}-\frac{pab}{2(a^2+b^2)}<P_4(a,b)=\frac{1}{4}+\frac{p(b^2-a^2)}{2(a^2+b^2)}.$
	Combining the above analysis and substituting $a=\cos u-1,b=\sin u$ into Eq.\eqref{Q1}, we get
	\begin{gather}
		\max\limits_{\vec{P}_{obs}}P(a,b|A_x,B_y)=
		\frac{1+p\cos u}{4}	,
	\end{gather}
	where $\cos u=\frac{\alpha}{\sqrt{\alpha^2+1}}$, and $R_{\min}=-\log_2(\frac{1+p\cos u}{4}),\ u\in(0,\frac{\pi}{4}],\ p\in(\frac{1}{\sqrt{2}},1]$. Through analysis we find that if we fix $\alpha$, the randomness decreases with the increase of p. On the contrary, if we fix $p$, the randomness decreases with the increase of $\alpha$. We fix $\alpha = 1$ and compare with the previous results. We plot the variable $p$ against the image of the maximum joint probability as in FIG. \ref{P2}. The randomness is found to reach a minimum value of 1.2884 bits for $p=1$, which is the same as that in \cite{ref6}.
	
	In \cite{ref11}, the complete measurement statistics are considered to represent randomness. And SDP is sketched  by running through all the quantum states and measurements that can reach the maximal violation to obtain the maximum value of the probability. We now fix the quantum states and measurements corresponding to the maximum violation. In \cite{ref11} the authors give an upper bound of the maximum joint probability, while here we give a lower bound, as shown and compared in FIG. \ref{P2}. And we can observe that when the value of $p$ is close to 1,  the upper and lower bounds meet.
	\begin{figure}[t]
		\centering
		{\includegraphics[width=0.5\textwidth]{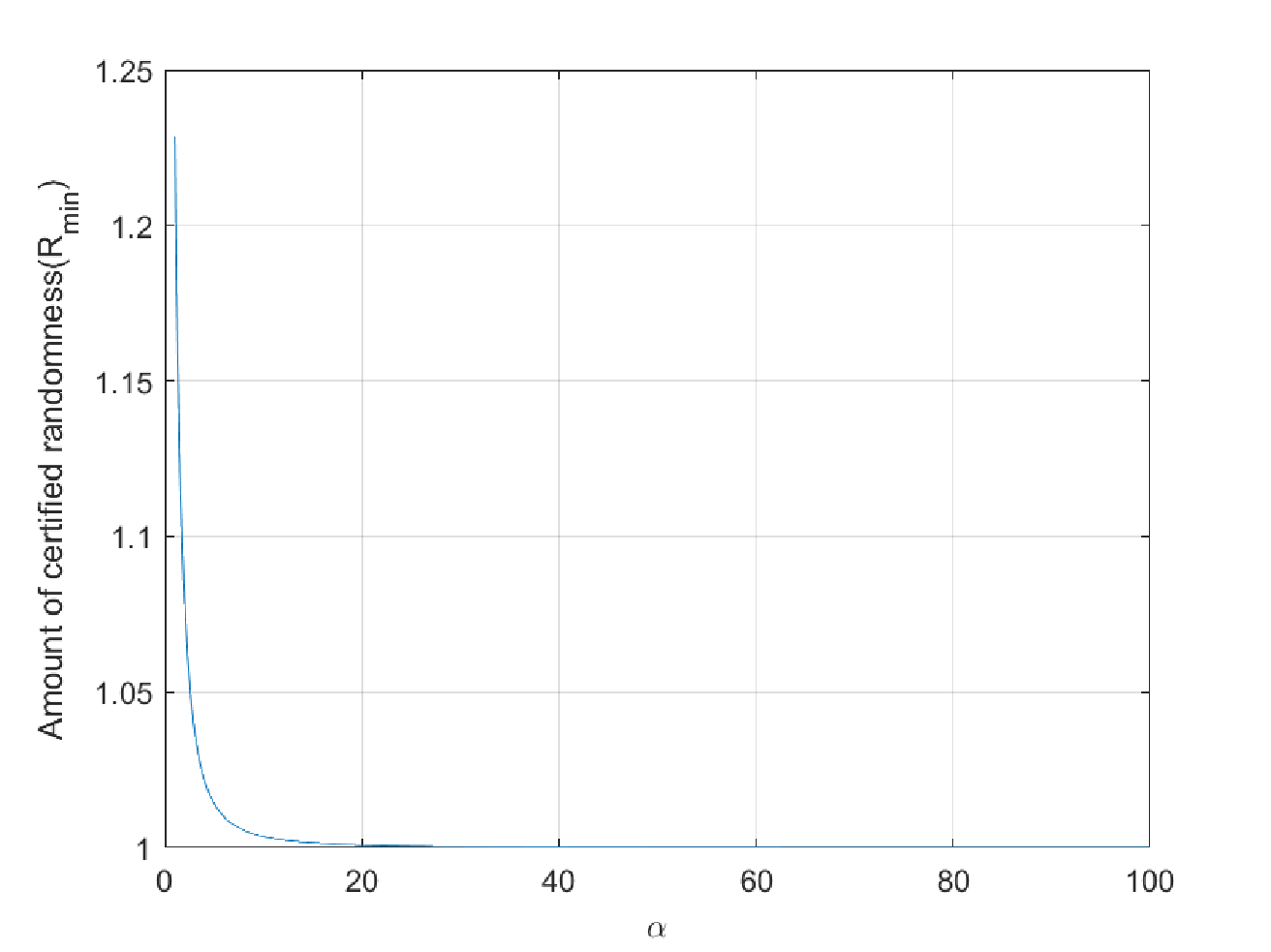}}\\
		\caption{The amount of certified randomness $ R_{\min}$ as a function of the visibility $\alpha$ for optimally violating generalized CHSH-Bell inequality.}
		\label{P1}
	\end{figure}
	\begin{figure}
		{\includegraphics[width=0.5\textwidth]{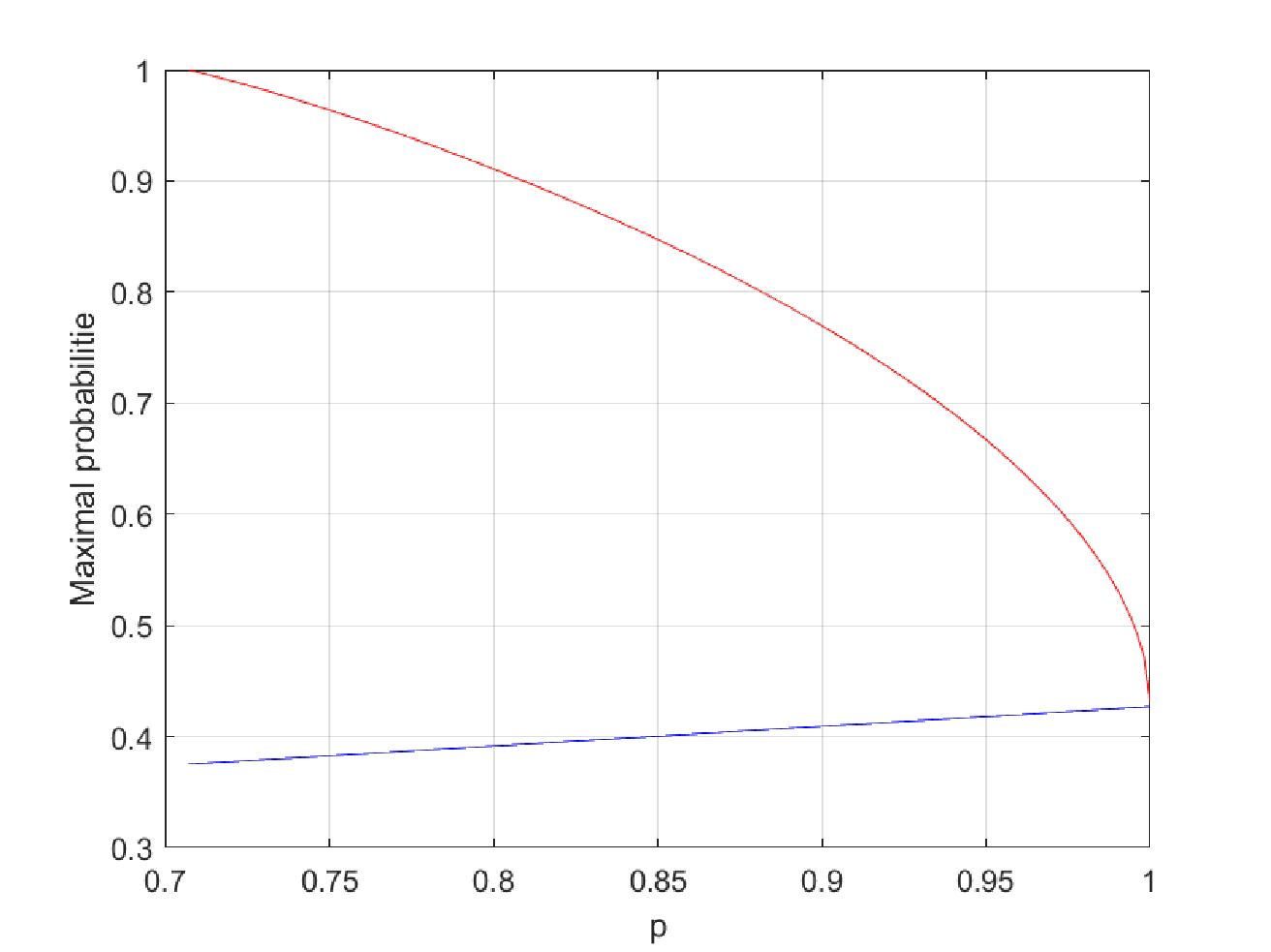}}\\
		\caption{Consider Werner state and the CHSH-Bell inequality. The maximal probability as a function of visibility $p$ for optimally violating CHSH correlations in present of white noise. The red line is the upper bound of the maximum joint probability obtained with the SDP program, while the blue line in the figure is a lower bound of the maximum joint probability obtained in this paper with the SOS decomposition. We find that the upper and lower bounds are close when $p$ is approaching to 1, and overlap when $p = 1$, with the corresponding maximum joint probability is $P_{\max}=\frac{1}{4}(1+\frac{1}{\sqrt{2}})$. The corresponding state is the maximally entangled state, which produces $R_{\min}$=1.2284 bits randomness.}
		\label{P2}
	\end{figure}
	
	\section{summary and outlook}
	In this paper, we present a kind of SOS decomposition form for general Bell inequalities and derive the measurement operators associated with the maximal violations of considered Bell inequalities. We also practice the SOS decomposition method by considering several examples. The corresponding SOS decomposition and the measurement operators that bring about the maximum violation values of these Bell inequalities are derived, which are consistent with previous results. In this paper we only consider maximally entangled quantum states to obtain the optimal measurement operator, but for general quantum states our approach is equally applicable. In particular, for the SOS decomposition of the generalized CHSH-Bell inequality that we obtain, we quantify the amount of randomness by the minimal entropy. We compute the randomness corresponding to the quantum optimal violation of the generalized CHSH-Bell inequality when the quantum states are the two-qubit maximally entangled state and the Werner state, respectively. 	
	For the maximally entangled state, we obtain a functional representation of the maximum guess probability with respect to the variable $\alpha$ as a function of $P_{\max}$ when the generalized CHSH-Bell inequality reaches the optimal violation. For Werner state, we obtain the maximum guessing probability as a function of the variable $\alpha$ and the white noise quantization value $p$ for the same Bell inequality violated maximally. In particular, we analyze the degeneration of the corresponding Bell inequality to CHSH-Bell inequality when $\alpha =1$. We obtain an upper bound on the maximum guessing probability using SDP and a lower bound on the maximum guessing probability using the SOS decomposition constructed in this paper, respectively. We find that the upper and lower bounds are close when $p$ is approaching to 1, which provides an effective estimation of the value of maximum guessing probability. 
	
	For further investigation, one can improve the lower bound of maximum guessing probability by performing local unitary operations on the measurements derived by SOS decomposition. The SOS decomposition method can be also discussed for multi-partite qubit systems.
	
\section*{Acknowledgments}
This work is supported by the Shandong Provincial Natural Science Foundation for Quantum Science No. ZR2021LLZ002 and the Fundamental Research Funds for the Central Universities No. 22CX03005A.

\section*{\bf Data availability statement}
No data was used for the research described in the article.

\end{document}